\documentclass[14pt,a4paper]{article}
\usepackage[utf8]{inputenc}
\usepackage[english]{babel}
\usepackage{amsmath}
\usepackage{amsfonts}
\usepackage{amssymb}
\usepackage{amsthm}
\usepackage{appendix}
\usepackage{cite}
\usepackage{graphicx}
\usepackage{graphicx,overpic}
\usepackage{subfigure}
\usepackage{changes}
\usepackage[left=2cm,right=2cm,top=2cm,bottom=2cm]{geometry}
\newtheorem{te}{Theorem}

\providecommand{\keywords}[1]
{\small	\textbf{\textit{Keywords:}} #1}

\begin{document}

\author{A. M. Escobar-Ruiz, R. Azuaje and J. C. Gordiano\\
Departamento de F\'{i}sica, Universidad Aut\'onoma Metropolitana Unidad Iztapalapa,\\ San Rafael Atlixco 186, 09340, Ciudad de M\'exico, M\'exico}

\title{Two-dimensional classical superintegrable systems: polynomial algebra of integrals}
\maketitle

\begin{abstract}

In this work, we investigate generic classical two-dimensional (2D) superintegrable Hamiltonian systems \( \mathcal{H} \), characterized by the existence of three functionally independent integrals of motion \( (\mathcal{I}_0 = H,\, \mathcal{I}_1,\, \mathcal{I}_2) \). Our main result, formulated and proved as a theorem, establishes that the set \( (\mathcal{I}_0,\, \mathcal{I}_1,\, \mathcal{I}_2,\, \mathcal{I}_{12} = \{\mathcal{I}_1,\, \mathcal{I}_2\}) \) generates a four-dimensional polynomial algebra under the Poisson bracket. Unlike previous studies, this study describes a construction that neither depends on the additive separability of the Hamilton–Jacobi equation nor presupposes polynomial integrals of motion in the canonical momenta. Specifically, we prove an instrumental observation presented in [D. Bonatsos et al., PRA 50, 3700 (1994)] concerning deformed oscillator algebras in  superintegrable systems. We apply the method to a variety of physically relevant examples, including the Kepler system, Holt potential, Smorodinsky–Winternitz potential, Fokas–Lagerstrom potential, the Higgs oscillator, and the non-separable Post–Winternitz system. In several cases, we explicitly derive the form of the classical trajectories \( y = y(x\,;\,\mathcal{I}_0,\mathcal{I}_1,\mathcal{I}_2) \) using purely algebraic means. Moreover, by examining the conditions under which \( \mathcal{I}_1 = \mathcal{I}_2 = 0 \), we identify and characterize special classes of trajectories.


\end{abstract}

\section{Introduction}
\label{sec1}

In classical mechanics, a Hamiltonian system with \( n \) degrees of freedom is said to be Liouville integrable if it admits \( n \) functionally independent integrals of motion \( \mathcal{I}_i \) that are in mutual involution, i.e., \( \{ \mathcal{I}_i, \mathcal{I}_j \} = 0 \) for all \( i, j = 1, \ldots, n \). A system is termed \textit{superintegrable} if it possesses more than \( n \) independent integrals of motion, and \textit{maximally superintegrable} when the total number reaches \( 2n - 1 \), the theoretical maximum. The presence of these additional conserved quantities imposes strong constraints on the system dynamics and often allows for exact solvability. Classical examples of maximally superintegrable systems include the $n$-dimensional Kepler potential \( V = \alpha/r \) and the radial harmonic oscillator \( V = \omega^2 r^2 \), both of which admit \( 2n - 1 \) explicit constants of motion. For a comprehensive treatment of superintegrable systems in both classical and quantum contexts, we refer the reader to the review article \cite{miller2013classical}.

In two dimensions (\( n = 2 \)), maximally superintegrable systems admit three functionally independent integrals of motion \( (\mathcal{I}_0 = H, \mathcal{I}_1, \mathcal{I}_2) \), with a fourth derived object \( \mathcal{I}_{12} = \{ \mathcal{I}_1, \mathcal{I}_2 \} \) that is generally non-vanishing. For a wide class of such systems, it has been shown that the set \( (\mathcal{I}_0, \mathcal{I}_1, \mathcal{I}_2, \mathcal{I}_{12}) \) forms a closed, four-dimensional polynomial algebra under the Poisson bracket \cite{daskaloyannis2001quadratic,daskaloyannis2006unified,marquette2007polynomial,kress2007equivalence,marquette2010superintegrability,post2011models,miller2013classical}. These studies have primarily focused on second-order superintegrable systems, where all integrals are at most quadratic in the canonical momenta. Notably, Daskaloyannis \cite{daskaloyannis2001quadratic} established the general structure of the corresponding quadratic Poisson algebra. This construction was later formalized in \cite{daskaloyannis2006unified}, where the existence of such an algebra—previously assumed evident—was rigorously demonstrated for all second-order superintegrable systems on two-dimensional manifolds. This polynomial algebra framework has since been used to classify superintegrable systems and explore connections to symmetry algebras \cite{daskaloyannis2007quantum,kalnins2005second,kress2007equivalence}, and has also been extended to systems with higher-order integrals and higher-dimensional settings \cite{marquette2007polynomial,marquette2009superintegrability,marquette2010superintegrability,daskaloyannis2010quadratic,marquette2013quartic}.

To the best of our knowledge, no general proof exists in the literature ensuring the existence of a polynomial algebra of integrals for arbitrary two-dimensional superintegrable systems, particularly when the integrals are not assumed to be polynomial in the momenta or are of order higher than two.

A significant conjecture in this direction was proposed in \cite{bonatsos1994deformed}, suggesting that for any two-dimensional polynomial superintegrable Hamiltonian system, it is always possible to construct two independent integrals \( \mathcal{I}_1 \) and \( \mathcal{I}_2 \) such that their nested Poisson bracket satisfies \( \{ \mathcal{I}_1, \mathcal{I}_{12} \} = -\mathcal{I}_2 \), where \( \mathcal{I}_{12} = \{ \mathcal{I}_1, \mathcal{I}_2 \} \). This structure closely resembles that of a deformed oscillator algebra, and has been employed to analyze quantum superintegrable systems through algebraic methods.

It is worth noting that an open question in the literature concerns the role of polynomial Poisson algebras in the definition of superintegrability. Specifically, it remains unclear whether the existence of a closed polynomial algebra—generated by functions on phase space and closed under the Poisson bracket—is sufficient to ensure that these functions are integrals of motion. In other words, \textit{under what conditions does algebraic closure imply that all generators commute with the Hamiltonian?} Although this inverse problem is not addressed in the present work, it was explored in~\cite{Yurducsen2021DoublyEN}, highlighting its potential relevance to the algebraic characterization of superintegrable systems.

In the present work, we provide a rigorous and constructive proof of an equivalent formulation of this conjecture. Our result establishes that a four-dimensional polynomial algebra of integrals of motion exists for any two-dimensional classical superintegrable Hamiltonian system, without assuming that the integrals are polynomial in the momenta or of limited order. This generalization significantly broadens the applicability of polynomial algebra techniques. Furthermore, we present explicit realizations of this algebra for several physically relevant systems, including cases where the integrals are non-polynomial, thereby demonstrating the utility and scope of our method.

Polynomial algebras play a central role not only in classical mechanics but also in the quantum realm, where their structure is closely tied to the concept of superintegrability \cite{miller2013classical}. In quantum systems, the existence of additional integrals of motion that do not commute among themselves leads naturally to non-Abelian symmetry algebras, which in many notable cases take the form of finite-order polynomial algebras---quadratic, cubic, or higher \cite{daskaloyannis2001quadratic, BONATSOS95, polynomial_latini_2022} (and references therein). These algebras generalize Lie algebras and encode essential information about the underlying dynamics and symmetries of the system. Their presence allows for a powerful algebraic approach to spectral analysis: the energy levels and degeneracies of the system can often be derived without solving the Schrödinger equation directly, by constructing finite-dimensional representations of the algebra and imposing appropriate quantization conditions \cite{marquette2010superintegrability}. This is typically achieved through a realization in terms of a \emph{deformed oscillator algebra} \cite{BONATSOS95, boson_ruan_1999,  algebraic_tanoudis_2011}, where the commutation relations are expressed using creation and annihilation operators satisfying
\[
[a, a^\dagger] = \Phi(N + 1) - \Phi(N),
\]
with \( N \) the number operator. The function \( \Phi(N) \), known as the \emph{structure function}, encodes the specific features of the polynomial algebra, including the Casimir operator and model parameters, and determines the admissible spectrum by enforcing positivity and truncation conditions on the representation space, namely
\[
\Phi(0) = 0, \quad \Phi(n + 1) = 0, \quad \Phi(k) > 0 \quad \text{for } 1 \leq k \leq n.
\]
As such, polynomial algebras provide not only a unifying language for classifying and analyzing quantum superintegrable systems but also offer valuable insights that motivate analogous investigations in the classical setting. An accessible and up-to-date discussion of polynomial algebras in the TTW system appears in \cite{vieyra2025} (see also \cite{10.1063/5.0201981, Turbiner_2023}).

\subsection{Generalities}

As a first step, let us recall the basics on the formulation of Hamiltonian Mechanics (for more details see References \cite{AKN2006, BBT2003, torres2018introduction}).
 
Let $M$ be a smooth manifold. A symplectic structure on $M$ is a closed non-degenerate 2-form $\omega$ on $M$. Closed means that $d\omega=0$ and non-degenerate implies that for each 1-form $\alpha$ on $M$ there is one and only one vector field $X$ which obeys $X \lrcorner \omega=\alpha$. A symplectic manifold is a pair $(M,\omega)$ where $\omega$ is a symplectic structure on $M$. By definition, each symplectic manifold is of even dimension.

Let $(M,\omega)$ be a symplectic manifold of dimension $2n$. Around any point $p\in M$ there exist local coordinates $(q^{1},\cdots,q^{n},p_{1},\cdots,p_{n})$, called canonical coordinates or Darboux coordinates, such that
\begin{equation}
\omega\ = \ dq^{i}\wedge dp_{i}\ .
\end{equation}
In this paper, we adopt the Einstein summation (\textit{i.e.}, a summation over repeated indices is assumed). In Classical Mechanics, the (contravariant) variables $\{ q^{n} \}$ are called generalized coordinates whilst the (covariant) quantities $\{ p_{n} \}$ correspond to their canonical momenta, respectively.   

For each function $f(p,q)\in C^{\infty}(M)$ in the phase space, is assigned a vector field $X_{f}$ on $M$, called the Hamiltonian vector field for $f$, according to
\begin{equation}
X_{f}\lrcorner \omega \ = \ df \ .
\end{equation}
In canonical coordinates $(q,p)$, this vector field $X_{f}$ takes the form
\begin{equation}
X_{f}=\frac{\partial f}{\partial p_{i}}\frac{\partial}{\partial q^{i}}-\frac{\partial f}{\partial q^{i}}\frac{\partial}{\partial p_{i}}\ .
\end{equation}
The assignment $f\longmapsto X_{f}$ is linear, namely
\begin{equation}
X_{f+\alpha \,g}\ = \ X_{f}+\alpha\, X_{g}\ ,
\end{equation}
$\forall\, f,g\in C^{\infty}(M)$ and $\forall\, \alpha \in\mathbb{R}$. For given functions $f,g \in C^{\infty}(M)$, the Poisson bracket of $f$ and $g$ is defined by
\begin{equation}
\label{PB}
\lbrace f,g\rbrace\ = \ X_{g}\,f \ = \ \omega(X_{f},X_{g})\ .
\end{equation}
In canonical coordinates, we have
\begin{equation}
\label{CPB}
\lbrace f,g\rbrace \ = \ \frac{\partial f}{\partial q^{i}}\frac{\partial g}{\partial p_{i}}\,-\,\frac{\partial f}{\partial p_{i}}\frac{\partial g}{\partial q^{i}} \ .
\end{equation}

The geometric theory of conservative Hamiltonian systems, when the Hamiltonian $H$ is not an explicit function of time $ t $, is naturally constructed within the mathematical formalism of symplectic geometry. Given $H\in C^{\infty}(M)$ the dynamics of the Hamiltonian system on $(M,\omega)$ (the phase space) is governed by the Hamiltonian vector field $X_{H}$. In this case, we say that $(M,\omega,H)$ is a Hamiltonian system with $n$ degrees of freedom, and the trajectories of the system $\psi(t)=(q^{1}(t),\cdots,q^{n}(t),p_{1}(t),\cdots,p_{n}(t))$ are the integral curves of $X_{H}$. They satisfy the Hamilton's equations of motion
\begin{equation}
\dot{q^{i}} =\frac{\partial H}{\partial p_{i}}, \hspace{1cm}
\dot{p_{i}} =-\frac{\partial H}{\partial q^{i}}\qquad ;\qquad i=1,2,3,\ldots,n \ .
\end{equation}

The evolution (the temporal evolution) of a function $f\in C^{\infty}(M)$ (a physical observable) along the trajectories of the system is given by
\begin{equation}
\dot{f}= \mathcal{L}_{X_{H}}f=X_{H}f=\lbrace f,H\rbrace \ ,
\end{equation}
where $\mathcal{L}_{X_{H}}f$ is the Lie derivative of $f$ with respect to $X_{H}$. A function $f$ is a global constant of motion of the system (or an integral of motion) if it is constant along all the trajectories of the Hamiltonian system, that is, $f$ is a constant of motion if $\mathcal{L}_{X_{H}}f=0$ (or equivalently $\lbrace f,H\rbrace=0$). Clearly, a global integral is conserved for any set of initial conditions.

The existence of such a global integrals of motion can be used to systematically reduce the number of degrees of freedom. Thus, the task of solving the corresponding Hamilton's equation of motion becomes simpler. For instance, the aforementioned symmetry reduction is exploited in the alternative Hamilton-Jacobi theory which is based on selecting the global integrals of motion as the new momenta. Eventually, the Hamiltonian is transformed to a suitable form such that the Hamilton's equations become directly integrable.

\section{Polynomial algebra of constants of motion}

Let $(M,\omega,H)$ be a classical superintegrable autonomous Hamiltonian system with two degrees of freedom ($n=2$). By definition, there exist two additional functions $F,G\in C^{\infty}(M)$ such that $H,F,G$ are three functionally independent conserved quantities under time evolution. Equivalently, they Poisson commute with $H$, $\lbrace  F,\,H\rbrace = \lbrace  G,\,H\rbrace =0 $ but not between themselves $R\equiv \lbrace  F,\,G\rbrace \neq 0$. 
Any other constant of motion of the system is necessarily a function of $H,F,G$. Consequently, $(H, F,G, R)$ generates a Poisson algebra, which we refer to as the algebra of integrals of motion.

\subsection{Main result of the study}

In \cite{bonatsos1994deformed}, it was conjectured that we can always construct two integrals of motion 
\begin{equation}
{\cal I}_1 \ = \ {\cal I}_1(H,F,G)\quad \text{and} \quad {\cal I}_2\ = \ {\cal I}_2(H,F,G)    \ ,
\end{equation}
such that 
\begin{equation}
\lbrace  {\cal I}_1,\,{\cal I}_{12}\rbrace \ = \ -{\cal I}_{2} \ ,
\end{equation}
where $\lbrace {\cal I}_1,\,{\cal I}_2\rbrace \, \equiv \, {\cal I}_{12} \,$. Below, we present a rigorous proof of this statement. This property is, in fact, the key element that ensures the resulting algebra of integrals is both closed and polynomial in nature. Hereafter, for a comparative purpose, it is convenient to adopt the notation used in \cite{bonatsos1994deformed}, namely 

\begin{equation}
{\cal I}_1 = L \quad , \quad {\cal I}_2 = A \quad , \quad  {\cal I}_{12} = B \ . 
\end{equation}

\begin{te}
\label{tetwoconstants}
If $(M,\omega,H)$ is a classical superintegrable Hamiltonian system with two degrees of freedom, one can always construct two constants of motion $L$ and $A$ such that 
\begin{equation}
\label{eqbonatsos}
\qquad \lbrace L,B\rbrace \ = \ -k^{2}A \ , \qquad k\in \mathbb{R} \ ,
\end{equation}
here $\lbrace L,A\rbrace \equiv B $.
\end{te}
\begin{proof}
First we rewrite (\ref{eqbonatsos}) as follows
\begin{equation}
\lbrace L,\lbrace L,A\rbrace\rbrace \ = \ -k^{2}\,A \ ,
\end{equation}
which is equivalent to 
\begin{equation}
X_{L}X_{L}A \ = \ -k^{2}A \ ,
\end{equation}
where $X_{L}$ is the Hamiltonian vector field for $L$. Therefore, finding the desired integrals of motion $L$ and $A$ reduces to the equations 
\begin{equation}
\label{eqbonatsos1}
X_{H}\,A \ = \ 0 \qquad\text{and}\qquad X_{L}X_{L}A \ = \ -k^{2}\,A \ .
\end{equation}

Now, since the system is superintegrable it is also Liouville integrable. This implies that there exists a constant of motion $L$ such that $H$ and $L$ are functionally independent. By Hamilton-Jacobi theory, it is always possible to find a canonical transformation of the form $(q_{1},q_{2},p_{1},p_{2})\mapsto(Q_{1},Q_{2},P_{1},P_{2})$ for which the new momentum coordinates are given by $P_{1}=H$ and $P_{2}=L$ \cite{AKN2006,torres2018introduction,BBT2003,GPS2002}. Thus, $Q_{2}$ is a locally defined integral of motion\footnote{ $Q_{2},H,L$ are functionally independent quantities on the domain where $Q_{2}$ is well defined.} that, due to the superintegrability, may be global. The local expressions for the Hamiltonian vector fields $X_{H}$ and $X_{L}$ in the canonical coordinates $(Q_{1},Q_{2},P_{1},P_{2})$ read
\begin{equation}
X_{H}=\frac{\partial}{\partial Q_{1}} \qquad\text{and}\qquad X_{L}=\frac{\partial}{\partial Q_{2}} \ .
\end{equation}
So (\ref{eqbonatsos1}) takes the local form
\begin{equation}
\label{eqbonatsoslocal}
\frac{\partial}{\partial Q_{1}}A\ = \ 0 \qquad\text{and}\qquad \frac{\partial^{2}}{\partial Q_{2}^{2}}\,A\ = \ -k^{2}\,A \ ,
\end{equation}
which is easily solved by choosing
\begin{equation}
\label{eqgeneralform}
A \ = \ A(Q_{2},H,L) \ = \ \sin(k\,Q_{2})\,F_{1}(H,L) \ + \ \cos(kQ_{2})\,F_{2}(H,L)\ ,
\end{equation}
where $F_{1}$ and $F_{2}$ are arbitrary functions depending only on $H$ and $L$.
\end{proof}

Next, let us suppose that $L$ and $A$ are two constants of motion satisfying equation (\ref{eqbonatsos}). Since $A=A(Q_2,\,P_{1}=H,\,P_{2}=L)$, we obtain
\begin{equation}
\frac{\partial A}{\partial Q_{1}}=\lbrace A,H\rbrace\ = \ 0 \ ,
\end{equation}
so $A=A(Q_{2},H,L)$, and
\begin{equation}
B\ = \ \lbrace L,A\rbrace=-\frac{\partial A}{\partial Q_{2}} \ . 
\end{equation}
Therefore, we can write 
\begin{equation}
B^2\ + \ k^{2}A^2\ = \ G(Q_{2},H,L)\ ,
\end{equation}
with $G$ a function depending, in general, on $Q_{2},H,L$. However, the straightforward computation   
\begin{equation}
\begin{split}
\frac{\partial G}{\partial Q_{2}} & \ = \  2\,B\,\frac{\partial B}{\partial Q_{2}}\ + \ 2\,k^{2}\,A\,\frac{\partial A}{\partial Q_{2}}\\
& \ = \  -2\,\frac{\partial A}{\partial Q_{2}}\frac{\partial B}{\partial Q_{2}}\ + \ 2\,\lbrace B,L \rbrace\,\frac{\partial A}{\partial Q_{2}}\\
&\ = \  -2\,\frac{\partial A}{\partial Q_{2}}\frac{\partial B}{\partial Q_{2}}\ + \ 2\,\frac{\partial B}{\partial Q_{2}}\frac{\partial A}{\partial Q_{2}}\\
&\ = \ 0 \ ,
\end{split}
\end{equation}
shows that
\begin{equation}
\label{BAred}
B^2\ + \ k^{2}\,A^2\ = \ G(H,L)\ .
\end{equation}
Thus, $G(H,L)$ is always a positive function and it vanishes only when both $A$ and $B$ are zero as well. 
From (\ref{BAred}), we obtain
\begin{equation}
\lbrace A, B^{2}\rbrace \ + \ k^{2}\,\lbrace A, \,A^{2}\rbrace\ = \ \lbrace A,\,G(H,L)\rbrace \ ,
\end{equation}
hence
\begin{equation}
2\,B\,\lbrace A, \,B\rbrace \ = \ \frac{\partial G}{\partial H}\lbrace A, H\rbrace \ + \ \frac{\partial G}{\partial L}\,\lbrace A, L\rbrace \ ,
\end{equation}
and
\begin{equation}
2B\lbrace A, B\rbrace\ = \ -\frac{\partial G}{\partial L}B \ .
\end{equation}
Eventually, we arrive to the expression
\begin{equation}
\label{ABG}
\lbrace A, B\rbrace \ = \ -\frac{1}{2}\frac{\partial G}{\partial L}\ ,
\end{equation}
in complete agreement with \cite{bonatsos1994deformed}.

\clearpage

\textbf{Remarks}

\begin{itemize}

\item Consider the case where the Hamiltonian \( H \) and the constants of motion \( L \) and \( A \), satisfying equation~\eqref{eqbonatsos}, are all polynomial functions of the momentum coordinates \( (p_1, p_2) \). From equations~\eqref{BAred} and~\eqref{ABG}, it follows that both \( G(H, L) \) and its partial derivative \( \partial G / \partial L \) must also be polynomial functions in the momenta. Consequently, since \( G \) depends only on \( H \) and \( L \), and both are polynomial in the momenta, we conclude that \( G(H, L) \) must be a polynomial function in \( H \) and \( L \). To clarify, equation~\eqref{BAred} shows that \( G \) is polynomial in the momentum coordinates. This implies that it can be expressed as a function of \( H \) and \( L \), potentially with rational exponents. However, because both \( H \) and \( L \) are themselves polynomial in the momenta, and \( G \) is a polynomial when written in terms of those momenta, it follows that any rational exponents must in fact be integers. Similarly, equation~\eqref{ABG} implies that \( \partial G / \partial L \) is also polynomial in the momenta. This consistency ensures that the exponents of \( H \) and \( L \) appearing in \( G(H, L) \) must all be integers, thereby confirming that \( G \) is a true polynomial in \( H \) and \( L \).

\item Since \( G(H, L) \) is a two-variable polynomial in \( H \) and \( L \), the set of integrals \( (H, L, A, B) \) forms a \emph{4-generated polynomial algebra} under the Poisson bracket. The nonvanishing Poisson relations are given by
\[
\{L, A\} = B, \quad \{L, B\} = -k^2 A, \quad \{A, B\} = P(H, L),
\]
where \( P(H, L) \) is a polynomial function of \( H \) and \( L \), explicitly determined via equation~\eqref{ABG}. This algebraic structure provides the classical foundation for the study of corresponding quantum superintegrable systems, where analogous polynomial algebras play a central role in the algebraic derivation of the spectrum.

\item It is worth emphasizing that even in the most general case---where the Hamiltonian and the constants of motion are \emph{not necessarily polynomial} in the momentum coordinates---a polynomial algebra of constants of motion still arises. To the best of our knowledge, this scenario has not been explicitly addressed in the literature. Suppose we consider a superintegrable system with Hamiltonian \( H(q_1, q_2, p_1, p_2) \) and a constant of motion \( L(q_1, q_2, p_1, p_2) \), not assumed to be polynomial in the momenta. The proof of Theorem~\ref{tetwoconstants} guarantees the existence of a constant of motion \( A \) of the form~\eqref{eqgeneralform}, where \( F_1(H, L) \) and \( F_2(H, L) \) are arbitrary smooth functions. Then,
\[
B = \{L, A\} = -\frac{\partial A}{\partial Q_2} = -k \cos(k Q_2) F_1(H, L) + k \sin(k Q_2) F_2(H, L),
\]
and
\[
B^2 + k^2 A^2 = k^2 \left( F_1^2(H, L) + F_2^2(H, L) \right).
\]
Differentiating yields
\[
\{A, B\} = -k^2 \left( F_1(H, L) \frac{\partial F_1}{\partial L} + F_2(H, L) \frac{\partial F_2}{\partial L} \right).
\]
Hence, if \( F_1(H, L) \) and \( F_2(H, L) \) are chosen to be \emph{polynomial functions} of \( H \) and \( L \), the full set \( (H, L, A, B) \) again closes a \emph{polynomial Poisson algebra}. This result underscores the generality and robustness of the polynomial algebraic structure in classical superintegrable systems, extending beyond the traditionally studied polynomial Hamiltonians.

\item Although analogous polynomial algebras are well established in the quantum setting through Casimir operators and deformed oscillator realizations, the classical construction presented here is far from trivial. In contrast to the quantum case, where the algebraic structure is often postulated or abstractly imposed to fit known representation theory, our classical approach derives the polynomial closure explicitly from first principles. The emergence of the function \( G(H, L) \) from the Poisson algebra of integrals of motion is not assumed but obtained constructively, providing a concrete and intrinsic characterization of the algebraic structure. To the best of our knowledge, this direct derivation of a closed classical polynomial algebra---without requiring polynomial expressions in the momenta or separability assumptions---has not been previously formulated in this generality.

\end{itemize}

\section{Examples}

To clarify these ideas, we present in this section explicit examples where the polynomial algebra is constructed through the function \( G(H,L) \) (\ref{BAred}), without performing direct computations of the Poisson brackets. These examples demonstrate how one can derive trajectories in configuration space purely by algebraic means, bypassing the need to solve differential equations.

\subsection{The isotropic harmonic oscillator}

We begin by considering the Hamiltonian of the two-dimensional isotropic harmonic oscillator
\begin{equation}
\label{HI}
H_{I} \ = \  \frac{1}{2\,m}(p_{x}^{2} \,+\, p_{y}^{2}) \ + \ \frac{m\,\omega^{2}\,(x^{2} + y^{2})}{2}\ ,
\end{equation}
where \( (x, y) \) are Cartesian coordinates on the plane, and \( (p_x, p_y) \) denote the corresponding canonical momenta. In this case, two well-known constants of motion are the angular momentum and the Fradkin tensor \cite{fradkin1965three,ballesteros2011superintegrable}
\begin{equation}
L=L_{z} = x\,p_{y} - y\,p_{x}\ , \quad\text{and}\quad A=S_{xy} = p_{x}\,p_{y} \ + \ m^{2}\,\omega^{2}\,x\,y \ ,
\end{equation}
respectively.
Its Poisson bracket read
\begin{equation}
\begin{split}
\{L,A\} &= -p_{x}^{2} + p_{y}^{2} \  - \  \omega^{2}\,m^{2}\,x^{2}\  + \  \omega^{2}\,m^{2}\,y^{2}\\
&= B\ ,\\
&\text{and}\\
\{L,B\} &= -4\,p_{x}\,p_{y}\ - \ 4\,\omega^{2}\,m^{2}\,x\ y \ , \\
&= - 4\,A \ ,
\end{split}
\end{equation}
therefore $k^{2} = 4$, cf. (\ref{eqbonatsos}). On the other hand, for (\ref{BAred}) we have
\begin{equation}
\begin{split}
G(H_I,L) &  = \  B^{2} \ + \  4\,A^{2}, \\
&= p_{x}^{4} + 2\,p_{x}^{2}\,p_{y}^{2} + p_{y}^{4} + 2\,\omega^{2}m^{2}x^{2}(p_{x}^{2} - p_{y}^{2}) + \omega^{4}m^{4}x^{4} + 8\,\omega^{2}m^{2}x\,y\,p_{x}\,p_{y}  \\ 
&\  - 2\,\omega^{2}m^{2}y^{2}(p_{x}^{2} - p_{y}^{2}) + 2\,\omega^{4}m^{4}x^{2}y^{2} + \,\omega^{4}m^{4}y^{4} \ , \\
&= \ (2\,m\,H_I)^{2}\ - \ \omega^{2}\,m^{2}\,(2\,L)^{2} \ ,
\end{split}
\end{equation}
which, in turn, leads to the quadratic polynomial relation between the integrals $(H_I,L,A,B)$
\begin{equation}
     B^{2} \ + \ 4\,A^{2} \ - \ (2\,m\,H_I)^{2}\ + \ \omega^{2}\,m^{2}\,(2\,L)^{2} \ = \ 0 \ .
\end{equation}
Finally, from $\{A,B\} = -\frac{1}{2}\frac{\partial G}{\partial L}$ in (\ref{ABG}), we obtain 
\begin{equation}
\lbrace A,B \rbrace \  = \  4\,\omega^{2}\,m^{2}\,L \ .
\label{ABI}
\end{equation}
Thus, the polynomial algebra generated by the four elements $(H_I,L,A,B)$ is, in fact, linear. 

\subsubsection{Trajectories in the configuration space}

By fixing the values of the three integrals of motion \( (H_I, L, A) \), one obtains a system of three algebraic equations that can be used to eliminate the momentum variables \( p_x \) and \( p_y \). As a result, the trajectories of the system in configuration space can be determined purely through algebraic manipulation, yielding an explicit relation of the form
\[
y\  = \ y(x;\, H_I, L, A)\ ,
\]
where the constants of motion enter as parameters. This remarkable feature is a direct consequence of the system’s maximal superintegrability, which ensures that the dynamics are fully constrained by the integrals of motion without the need to solve differential equations.

Putting $m=1$ and $\omega=1$, the defining algebraic equation of such trajectories takes the form

\begin{equation}
\label{Tr1}
 L^2 \left(\,4 \,A \,x\, y-2\, H_I \left(x^2+y^2\right)+x^4-2\, x^2 y^2+y^4\right)\ + \ \left(A \left(x^2+y^2\right)-2 \,H_I\,  x \,y\,\right)^2+L^4\ = \ 0   \ .
\end{equation}
It is a fourth-order polynomial equation in the variables $x,y$. Here, the values of the integrals are determined by the initial conditions $(x_0,y_0,p_{x0},p_{y0})$.

For the case $A=0$, the integrals $L$ and $B$ are in \textit{particular} involution \cite{Escobar-Ruiz_2024}. In two-dimensional superintegrable systems, the presence of three functionally independent integrals of motion—typically the Hamiltonian \( H \) and two additional conserved quantities \( F \) and \( G \)—ensures that the system is maximally constrained, with all bounded trajectories closed and the dynamics highly regular \cite{nekhoroshev1972action}. In general, the additional integrals \( F \) and \( G \) fail to be in involution, satisfying a nontrivial Poisson bracket relation that reflects the non-Abelian character of the associated symmetry algebra. Nevertheless, there may exist dynamically invariant submanifolds of phase space, or particular trajectories, along which the Poisson bracket \( \{F, G\} \) vanishes identically. This circumstance does not alter the global superintegrability of the system, but it signifies a local abelianization of the algebra of integrals. On these submanifolds, the flows generated by \( F \) and \( G \) commute, and the system admits a local Liouville-integrable structure. Such configurations are of particular interest, as they correspond to algebraically and geometrically distinguished sectors of the phase space where the symmetry structure simplifies. They often underlie the existence of separable coordinate systems, permit the explicit integration of the equations of motion, and provide insight into spectral degeneracies and singular orbit structure.

\begin{figure*}[h]
	\centering
	\subfigure[$L\,A\neq 0$]{\includegraphics[width=0.45\textwidth]{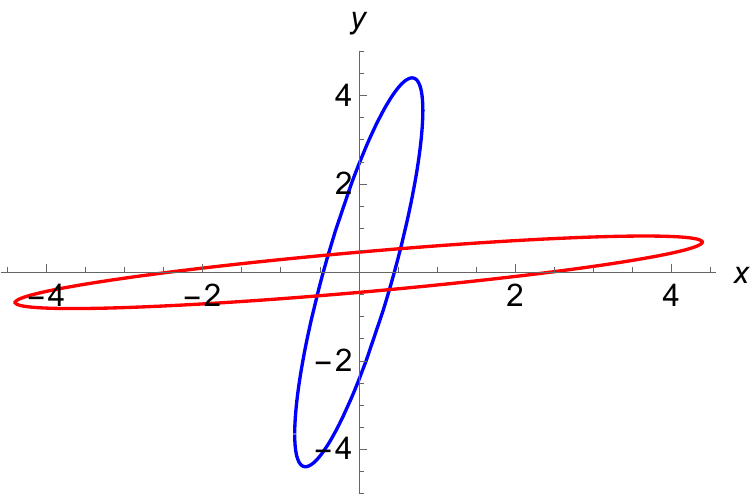}}\hfill
	\subfigure[$L= 0$, $A\neq0$]{\includegraphics[width=0.45
\textwidth]{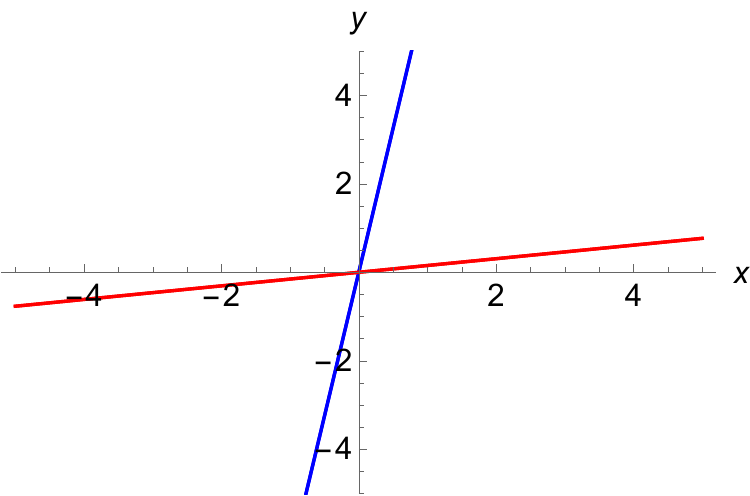}}\\
\subfigure[$A= 0$, $L\neq0$]{\includegraphics[width=0.45
\textwidth]{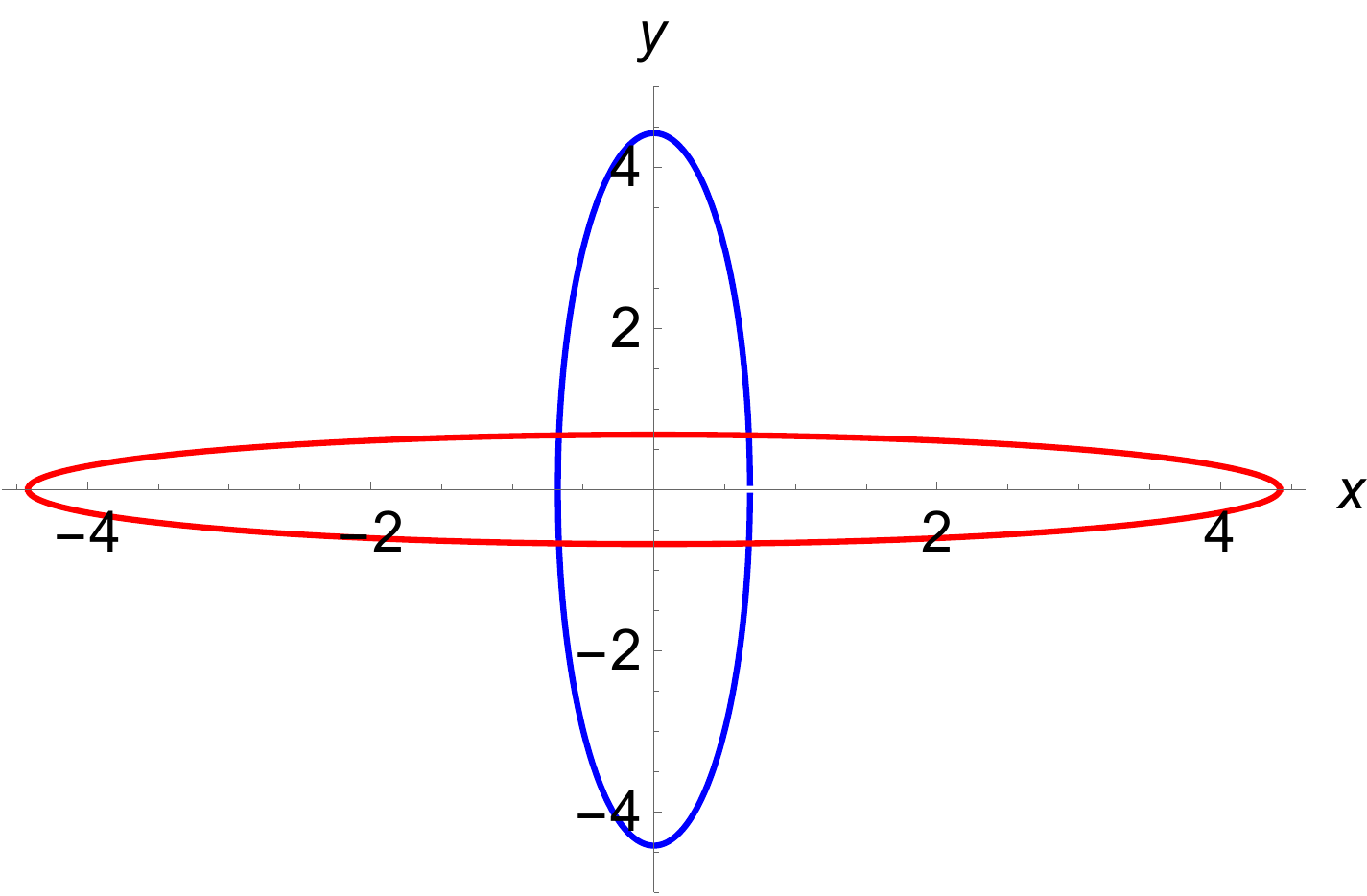}}
	\caption{Trajectories of the Hamiltonian $H_I$ (\ref{HI}), in the configuration space $(x,y)$, obtained from the algebraic equation (\ref{Tr1}). The plot (a) corresponds to the values $H_I=10,\,L=2,\,A=3$, (b) refers to the case $L= 0$ where $H_I=10,\,L=0,\,A=3$, whilst (c) displays the case $H_I=10,\,L=3,\,A=0$.}
	\label{H1F}
\end{figure*}

Figure~\ref{H1F} illustrates the trajectories of the Hamiltonian \( H_I \), as defined in equation~(28), in the configuration space \((x, y)\). These trajectories are obtained from the algebraic constraint equation~(34), without solving the equations of motion. In subfigure (a), corresponding to the values \( H_I = 10,\, L = 2,\, A = 3 \), all three integrals are nonzero, and the resulting trajectories form rotated ellipses, demonstrating coupling between degrees of freedom. Subfigure (b) shows the case \( H_I = 10,\, L = 0,\, A = 3 \), where the angular momentum \( L \) vanishes; the trajectories align along directions determined solely by the integral \( A \). In subfigure (c), with \( H_I = 10,\, L = 3,\, A = 0 \), the integral \( A \) vanishes, and the resulting trajectories form axis-aligned ellipses, corresponding to motion constrained along the principal directions of the system. These examples show how the form of the trajectories depends parametrically on the values of the integrals of motion, and how maximal superintegrability enables their determination through purely algebraic methods.

\subsection{The Kepler problem}

In this case, the Hamiltonian function reads
\begin{equation}
\label{HII}
H_{II} \ = \  \frac{1}{2}(\,p_{x}^{2} + p_{y}^{2}\,) \ - \  \frac{\alpha}{\sqrt{x^{2} + y^{2}}}\ , \qquad \alpha > 0\ .
\end{equation}
The integrals of motion are the angular momentum $L$ and the second component $A$ of the Laplace-Runge-Lenz vector \cite{miller2013classical}
\begin{equation}
L\ = \ x\,p_{y} - y\,p_{x}\ , \quad\text{and}\quad  A \ = \ - (x\,p_{y} - y\,p_{x})\,p_{x} - \frac{\alpha\, y}{\sqrt{x^{2} + y^{2}}}\ .
\end{equation}
We have
\begin{equation}
\begin{split}
\{L,A\} &= -x\,p_{y}^{2} \ + \  y\,p_{x}\,p_{y} \ + \  \frac{\alpha\, x}{\sqrt{x^{2} + y^{2}}}\\
&= \ B\ ,\\
\{L,B\} & \ = \ x\,p_{x}\,p_{y}\ - \ y\,p_{x}^{2} \ + \  \frac{\alpha \,y}{\sqrt{x^{2} + y^{2}}} \ ,\\
&= - A\ .
\end{split}
\end{equation}
therefore $k^{2} = 1$, cf. (\ref{eqbonatsos}). Now we write the function $G$ (\ref{BAred})
\begin{equation}
\begin{split}
G(H_{II},L) &=  \alpha^{2} \ + \  (x\,p_{y} - y\,p_{x})^{2}\,\left(p_{x}^{2} + p_{y}^{2} - \frac{2\alpha}{\sqrt{x^{2} + y^{2}}}\right)\\
&= \alpha^{2}\ + \ 2\,L^{2}\,H_{II} \ .
\end{split}
\end{equation}
Finally
\begin{equation}
\label{AB2}
\{A,B\}  \ = \  -2\,L\,H_{II} \ .
\end{equation}
Thus, unlike the previous case, the polynomial algebra generated by the four elements $(H_{II},L,A,B)$ is quadratic. This result is known in the literature; however, our rederivation follows a different line of reasoning and avoids direct Poisson bracket computations, relying instead on the structural properties of the integrals of motion.

\subsubsection{Trajectories in the configuration space}

Once again, by fixing the values of the three integrals of motion \( (H_{II}, L, A) \), we obtain a system of algebraic equations that allows for the elimination of the momentum variables \( p_x \) and \( p_y \). As a result, the trajectories of the system can be derived purely through algebraic manipulation, yielding an explicit expression of the form $y = y(x\,;\, H_{II}, L, A)$,
where the constants of motion enter as parameters.

Putting $m=1$ and $\alpha=1$, the defining equation of such trajectories is given by

\begin{equation}
\label{Tr2}
A^2 \left(x^2+y^2\right)\ + \ 2\, A\, y \sqrt{x^2+y^2}+L^4+y^2 \ - \ 2\, L^2 \,\left(A \, y+H_{II}\, x^2+\sqrt{x^2+y^2}\right)\ = \ 0   \ , 
\end{equation}
where $x^2+y^2 \neq0$\,. This algebraic equation can be reduced to a fourth-order polynomial equation in the variables $(x,y)$.
\begin{figure*}[h]
	\centering
	\subfigure[Case $L\, A\neq 0$]{\includegraphics[width=0.45\textwidth]{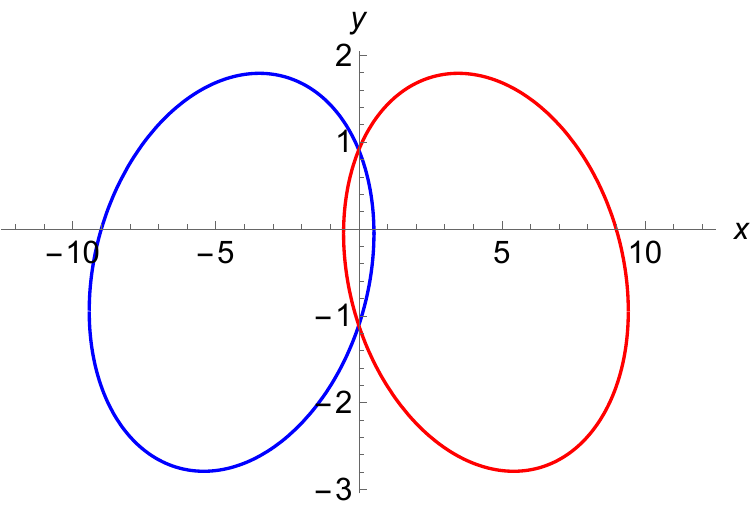}}\hfill
	\subfigure[$L= 0$, $A\neq 0$]{\includegraphics[width=0.45
\textwidth]{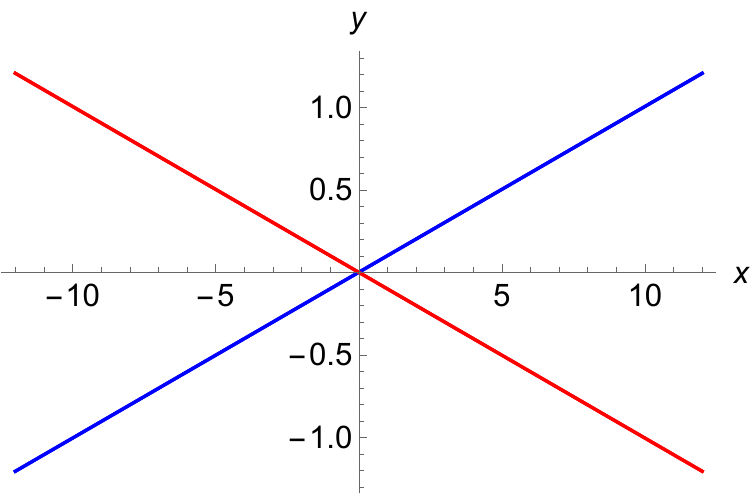}}
	\caption{Trajectories of the Hamiltonian $H_{II}$ (\ref{HII}), in the configuration space $(x,y)$, obtained from the algebraic equation (\ref{Tr2}). The plot (a) corresponds to the values $H_{II}=-\frac{1}{10},\,L=1,\,A=\frac{1}{10}$, (b) refers to the case $H_{II}=-\frac{1}{10},\,L=0,\,A=\frac{1}{10}$.}
	\label{H2F}
\end{figure*}

Figure~\ref{H2F} shows the trajectories of the Hamiltonian \( H_{II} \) in configuration space \( (x, y) \), obtained by eliminating the momenta using fixed values of the integrals \( (H_{II}, L, A) \). In both cases, the resulting algebraic equation defines a curve \( y = y(x) \), but due to its multivalued (typically quadratic) nature, \emph{two distinct trajectories} arise for each choice of constants. Subfigure (a), with \( L \neq 0 \) and \( A \neq 0 \), yields two closed curves, while subfigure (b), with \( L = 0 \), gives a pair of intersecting straight lines. This reflects the fact that fixing the integrals defines a level set in phase space, which can correspond to multiple geometrically distinct trajectories in configuration space.

\subsection{The Fokas-Lagerstrom system}

In this case, the Hamiltonian function is
\begin{equation}
\label{HIII}
H_{III} = \frac{1}{2}(p_{x}^{2} + p_{y}^{2}) + \frac{x^{2}}{2} + \frac{y^{2}}{18}.
\end{equation}
The constants of motion are given by \cite{bonatsos1994deformed}
\begin{equation}\label{eq2}
L = p_{x}^{2} + x^{2} \quad\text{and}\quad A = (xp_{y} - yp_{x})\,p_{y}^{2} + \frac{y^{3}\,p_{x}}{27} - \frac{x\,y^{2}p_{y}}{3}.
\end{equation}
We have
\begin{equation}
\begin{split}
\{L,A\} &= -2xyp_{y}^{2} + \frac{2xy^{3}}{27} - 2p_{x}p_{y}^{3} + \frac{2y^{2}p_{x}p_{y}}{3} \\
&=  B,\\
&\text{and}\\
\{L,B\} &=  -4\left((xp_{y} - yp_{x})p_{y}^{2} + \frac{y^{3}p_{x}}{27} - \frac{xy^{2}p_{y}}{3}\right) \\
&= -4\,A.
\end{split}
\end{equation}
Hence $k^{2} = 4$. On the other hand
\begin{equation*}
\begin{split}
B^{2} 
&= 4x^{2}y^{2}p_{y}^{4} - \frac{8x^{2}y^{4}p_{y}^{2}}{27} + \frac{4x^{2}y^{6}}{729} + 8xyp_{x}p_{y}^{5} - \frac{8xy^{3}p_{x}p_{y}^{3}}{3} - \frac{8xy^{3}p_{x}p_{y}^{3}}{27} + \\
&\quad + \frac{8xy^{5}p_{x}p_{y}}{81} + 4p_{x}^{2}p_{y}^{6} - \frac{8y^{2}p_{x}^{2}p_{y}^{4}}{3} + \frac{4y^{4}p_{x}^{2}p_{y}^{2}}{9},\\
&\text{and}\\
4\,A^{2} &= 4(xp_{y} - yp_{x})^{2}p_{y}^{4} + 8(xp_{y} - yp_{x})p_{y}^{2}\left(\frac{y^{3}p_{x}}{27} - \frac{xy^{2}p_{y}}{3}\right) + 4\left(\frac{y^{3}p_{x}}{27}- \frac{xy^{2}p_{y}}{3}\right)^{2} \\
&= 4x^{2}p_{y}^{6} - 8xyp_{x}p_{y}^{5} + 4y^{2}p_{x}^{2}p_{y}^{4} + \frac{8xy^{3}p_{x}p_{y}^{3}}{27} - \frac{8x^{2}y^{2}p_{y}^{4}}{3} - \frac{8y^{4}p_{x}^{2}p_{y}^{2}}{27} + \frac{8xy^{3}p_{x}p_{y}^{3}}{3} + \\
&\quad + \frac{4y^{6}p_{x}^{2}}{729} - \frac{8xy^{5}p_{x}p_{y}}{81} + \frac{4x^{2}y^{4}p_{y}^{2}}{9};
\end{split}
\end{equation*}
then
\begin{equation}
\begin{split}
G(H=H_{III},L) &= B^{2} + 4\,A^{2}\\
&= \frac{4x^{2}y^{2}p_{y}^{4}}{3} + \frac{4x^{2}y^{4}p_{y}^{2}}{27} + \frac{4x^{2}y^{6}}{729} + 4x^{2}p_{y}^{6} + 4p_{x}^{2}p_{y}^{6} + \frac{4y^{6}p_{x}^{2}}{729}  + \frac{4y^{2}p_{x}^{2}p_{y}^{4}}{3} +  \frac{4y^{4}p_{x}^{2}p_{y}^{2}}{27} \\
&= \frac{4}{3}y^{2}p_{y}^{4}(x^{2} + p_{x}^{2}) + \frac{4}{27}y^{4}p_{y}^{2}(x^{2} + p_{x}^{2}) + \frac{4}{729}y^{6}(x^{2} + p_{x}^{2}) + 4p_{y}^{6}(x^{2} + p_{x}^{2}) \\
&= (x^{2} + p_{x}^{2})\left(\frac{4y^{2}p_{y}^{4}}{3} + \frac{4y^{4}p_{y}^{2}}{27} + 4p_{y}^{6} + \frac{4y^{6}}{729}\right) \\
&= \frac{4}{729}(p_{x}^{2} + x^{2})(9p_{y}^{2} + y^{2})^{3} \\
&= -4L^{4} + 32H^{3}L - 48H^{2}L^{2} + 24HL^{3}.
\end{split}
\end{equation}
Eventually, we arrive to the Poisson bracket
\begin{equation}
\begin{split}
\{A,B\}= 8L^{3} - 16H^{3} + 48H^{2}L - 36HL^{2}\ .
\end{split}
\end{equation}

Thus, the polynomial algebra generated by the four elements $(H_{III},L,A,B)$ is cubic. 

\subsubsection{Trajectories in the configuration space}

From the conserved quantities \( (H_{III}, L, A) \), we can determine the trajectories of the system in configuration space. The equation defining these trajectories is given by

\begin{equation}
\label{Tr3}
\begin{aligned}
&531441 A^4+\big(-729 x^2 (2 H-L)^3+216 L y^4 (L-2 H)+729 L y^2 (L-2 H)^2+16 L y^6\big)^2
\\ & -1458 A^2 \left(-216 y^4 (2 H-L) \left(L-2 x^2\right)+729 y^2 (L-2 H)^2 \left(L-2 x^2\right)+729 x^2 (2 H-L)^3+16 y^6 \left(L-2 x^2\right)\right)  = 0 \ ,
\end{aligned}
\end{equation}
here $H=H_{III}$.

\begin{figure*}[h]
	\centering
 \subfigure[$L\, A\neq 0$]
{\includegraphics[width=0.4\textwidth]{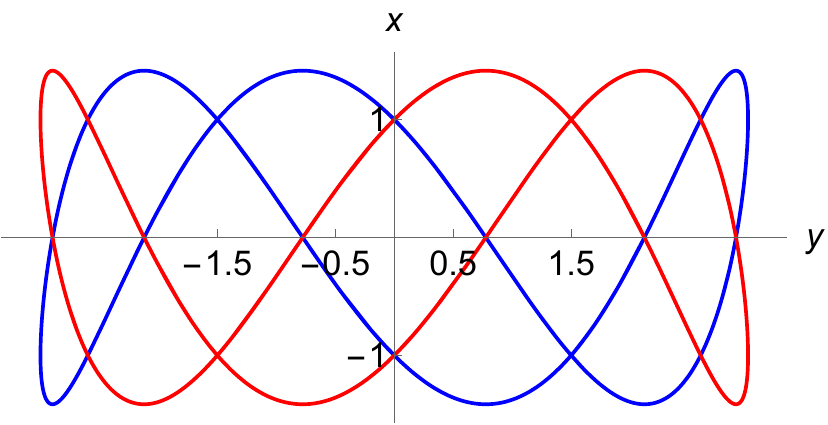}}\hfill
\subfigure[$A=0$, $L\neq 0$]
{\includegraphics[width=0.5\textwidth]{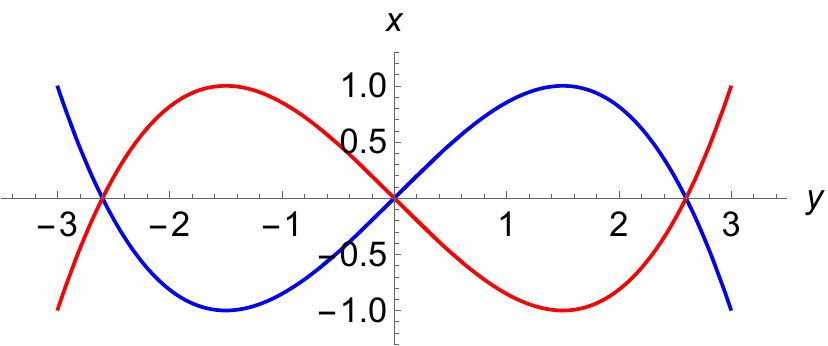}}
	\caption{Trajectories (blue and red lines) of the Hamiltonian $H_{III}$ (\ref{HIII}) obtained from the algebraic equation (\ref{Tr3}). The case (a) correspond to the values $H_{III}=\frac{3}{2},\,L=2,\,A=1$, whereas in (b) $H_{III}=1,\,L=1,\,A=0$. }
	\label{H3F}
\end{figure*}

Figure~\ref{H3F} presents the trajectories of the Hamiltonian \( H_{III} \), obtained from the algebraic constraint equation~(\ref{Tr3}), which relates \( x \) and \( y \) for fixed values of the conserved quantities \( (H_{III}, L, A) \). In this case, it is more convenient to represent the trajectories in the form \( x = x(y) \), rather than \( y = y(x) \), due to the structure of the defining equation. Subfigure (a), with nonzero values of both integrals \( L \) and \( A \), shows a pair of oscillatory trajectories (in blue and red) that exhibit a complex periodic structure. These curves are symmetric and highly structured, reflecting the interplay of both conserved quantities. In contrast, subfigure (b) corresponds to the degenerate case \( A = 0 \), where the algebraic curve simplifies, and the resulting trajectories take the form of lower-frequency periodic waves aligned along the \( y \)-axis. As in previous examples, fixing the integrals defines a multivalued algebraic constraint that gives rise to multiple geometrically distinct trajectories in configuration space.

\subsection{The Holt system}

In the following, we examine a more complex Hamiltonian that includes a singular term of the form \( 1/x^2 \).\cite{holt1982construction}
\begin{equation}
\label{HIV}
H_{IV} = \frac{1}{2}(p_{x}^{2} + p_{y}^{2}) + (x^{2} + 4y^{2}) + \frac{\delta}{x^{2}}\ ,
\end{equation}
$\delta \in \mathbb{R}$ is a parameter. The corresponding constants of motion quadratic in the momentum variables are  \cite{bonatsos1994deformed}
\begin{equation}
L = p_{y}^{2} + 8y^{2} \quad\text{and}\quad
A = p_{x}^{2}p_{y} + 8xyp_{x} - 2x^{2}p_{y} + \frac{2\delta p_{y}}{x^{2}}.
\end{equation}
We have
\begin{equation}
\begin{split}
\{L,A\} 
&= 16yp_{x}^{2} - 32x^{2}y + \frac{32y\delta}{x^{2}} - 16xp_{x}p_{y} \\
&= B,\\
&\text{and}\\
\{L,B\} 
&= -32\left(8xyp_{x} + p_{x}^{2}p_{y} - 2x^{2}p_{y} + \frac{2\delta p_{y}}{x^{2}}\right) \\
&= -32A.
\end{split}
\end{equation}
so $k^{2} = 32$. The function $G$ (\ref{BAred}) reads
\begin{equation}
\begin{split}
G(H=H_{IV},T) 
&= 32(8y^{2} + p_{y}^{2})\left(p_{x}^{4} + 4\,x^{2}\,p_{x}^{2} + 4\,x^{4} + \frac{4\,\delta \,p_{x}^{2}}{x^{2}} - 8\delta + \frac{4\,\delta^{2}}{x^{4}} \right)\\
&=32\,L\,(2H -L)^{2} \ - \ 512\,\delta\, L.
\end{split}
\end{equation} 
Finally
\begin{equation}
\begin{split}
\{A,B\} &=-16(2H - L)^{2} + 32L(2H -L) + 256\,\delta
\end{split}
\end{equation}

\noindent The polynomial algebra generated by the four elements $(H_{IV},L,A,B)$ is quadratic. 

\subsubsection{Trajectories in the configuration space}

From the system of integrals of motion, with $\delta=1$, the resulting algebraic equation governing the trajectories in configuration space is

\begin{equation}
\begin{aligned}
\label{Tr4}
 & A^4+\left(L \left(-2 \,H+L+4 x^2\right)^2-8 y^2 (-2 H+L-4) (-2 H+L+4)\right)^2
 \\ &
 \ - \ 2 A^2 \left(L \left(-2 H+L+4 x^2\right)^2-8 y^2 \left(16 x^2 (L-2 H)+(L-2 H)^2+32 x^4+16\right)\right) = 0  \ , 
\end{aligned}    
\end{equation}
$x\neq 0$, here $H=H_{IV}$, see Fig. \ref{H4F}.

\begin{figure*}[h]
	\centering
 \subfigure[$L\, A\neq 0$]
{\includegraphics[width=0.5\textwidth]{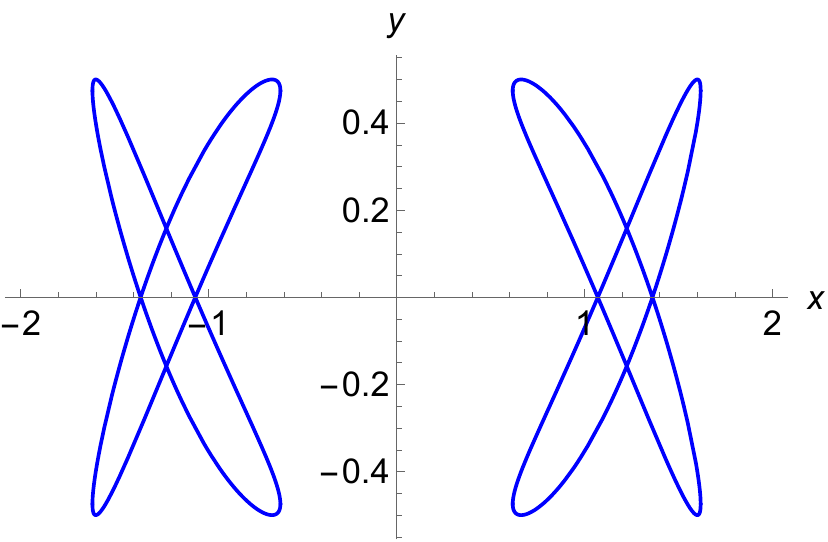}}\hfill
\subfigure[$A=0$, $L\neq 0$]
{\includegraphics[width=0.45\textwidth]{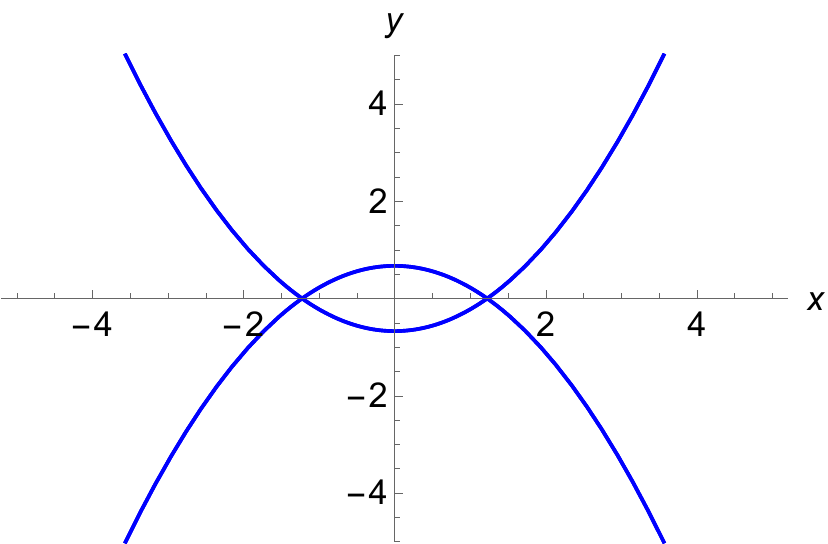}}
	\caption{Trajectories of the Hamiltonian $H_{IV}$ (\ref{HIV}) obtained from the algebraic equation (\ref{Tr4}). The case (a) correspond to the values $H_{IV}=4,\,L=2,\,A=2$, whereas in (b) $H_{IV}=4,\,L=2,\,A=0$. }
	\label{H4F}
\end{figure*}

\noindent As before, the topology of the trajectories is largely governed by the values of the conserved quantities.

\subsection{The Smorodinsky-Winternitz system}

The Smorodinsky-Winternitz Hamiltonian function, introduced in the 1960s, emerged as an extension of the isotropic harmonic oscillator with added inverse-square terms\cite{winternitz1967symmetry,bonatsos1994deformed}
\begin{equation}
\label{HV}
H_V = \frac{1}{2}(p_{x}^{2} + p_{y}^{2}) \ + \ b\,(x^{2} + y^{2}) \  + \ \frac{c}{x^{2}}\ ,
\end{equation}
$b,c \in \mathbb{R}$. The functionally independent constants of motion are the Hamiltonian $H_V$ and 
\begin{equation}
T \ = \ p_{y}^{2} \ + \ 2\,b\,y^{2} \quad\text{and}\quad
C \ = \ x^{2}\,p_{y}^{2} \ + \  y^{2}\,p_{x}^{2} \ - \ 2\,x\,y\,p_{x}\,p_{y} + \frac{2\,c\,y^{2}}{x^{2}}\ .
\end{equation}

We choose $L=T$ and $A=\lbrace T,C\rbrace=8bx^{2}yp_{y} - 8bxy^{2}p_{x} - 4yp_{x}^{2}p_{y} + 4xp_{x}p_{y}^{2} - \frac{8cyp_{y}}{x^{2}}$.
Therefore, we have
\begin{equation}
\lbrace L,A \rbrace  = -16b\left(x^{2}p_{y}^{2} + y^{2}p_{x}^{2} + \frac{2cy^{2}}{x^{2}} - 2xyp_{x}p_{y}\right) + 32\left(bxy + \frac{p_{x}p_{y}}{2}\right)^{2} + \frac{16cp_{y}^{2}}{x^{2}}=B\ ,
\end{equation}
and
\begin{equation}
\begin{split}
\lbrace L,B \rbrace &= 128byp_{x}^{2}p_{y} + \frac{256bcyp_{y}}{x^{2}} - 256b^{2}x^{2}yp_{y} - 128bxp_{x}p_{y}^{2} + 256b^{2}xy^{2}p_{x}, \\
&= -32\,b\,A.
\end{split} 
\end{equation}
Here $k^{2} = 32\,b$. The function $G$ (\ref{BAred}) takes the form
\begin{equation}
\begin{split}
G(H=H_V,L) &= \frac{64(4c^{2} + 4cx^{2}(p_{x}^{2} - 2bx^{2}) + x^{4}(p_{x}^{2} + 2bx^{2})^{2})(p_{y}^{2} + 2by^{2})^{2}}{x^{4}} , \\ 
&= 256\left(\left(H - \frac{L}{2}\right)^{2} - 4bc\right)L^{2}.
\end{split}
\end{equation}
Finally
\begin{equation}
\begin{split}
\lbrace A,B \rbrace &=  -\frac{1}{2}\frac{\partial G}{\partial L}, \\
&= -256 \left(\left(H - \frac{L}{2}\right)^{2} - 4bc\right)L + 128\left( H - \frac{L}{2} \right)L^{2} \ .
\end{split}
\end{equation}

Hence, the polynomial algebra generated by the four elements $(H_{V},L,A,B)$ is cubic. It is known that this system can be also characterized by a quadratic polynomial algebra of integrals \cite{miller2013classical}. Hence, it may happen that the present approach does not provide the \textit{minimal} polynomial algebra of integrals. Here, we mainly focus on the existence of the polynomial algebra itself. Moreover, as previously mentioned, our method is natural towards the study of the corresponding quantum systems by means of the deformed oscillator algebras.  

\subsubsection{Trajectories in the configuration space}

Taking $b=c=1$, the algebraic equation of the trajectories in the configuration space, obtained from the integrals of motion, reads

\begin{equation}
\begin{aligned}
\label{Tr5}
 & A^4+32 A^2 \bigg(L^2 \left(x^2 \left(-2 \,H+L+2 x^2\right)+2\right)+2 y^4 \left(16 x^2 (L-2 H)+(L-2 H)^2
 +32 x^4+16\right)
\\ & -L y^2 \left(16 x^2 (L-2 H)+(L-2 H)^2+32 x^4+16\right) \bigg)
 +256 \bigg(L^2 \left(x^2 \left(-2 H+L+2 x^2\right)+2\right)
\\ & 
 -2 y^4 (-2 H+L-4) (-2 H+L+4)+L y^2 (-2 H+L-4) (-2 H+L+4)\bigg)^2 = 0  \ , 
\end{aligned}    
\end{equation}

$x\neq 0$, here $H=H_{V}$, see Fig. \ref{H5F}.

\begin{figure*}[h]
	\centering
 \subfigure[$L\, A\neq 0$]
{\includegraphics[width=0.5\textwidth]{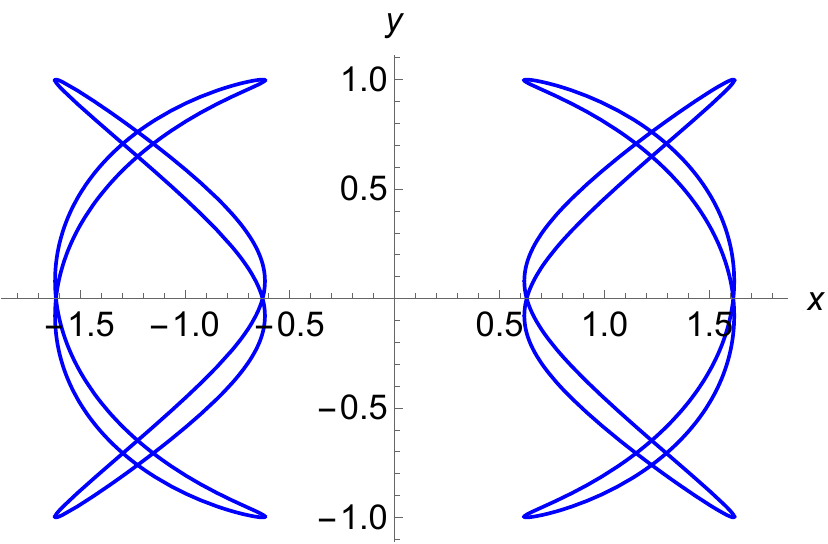}}\hfill
\subfigure[$A=0$, $L\neq 0$]
{\includegraphics[width=0.45\textwidth]{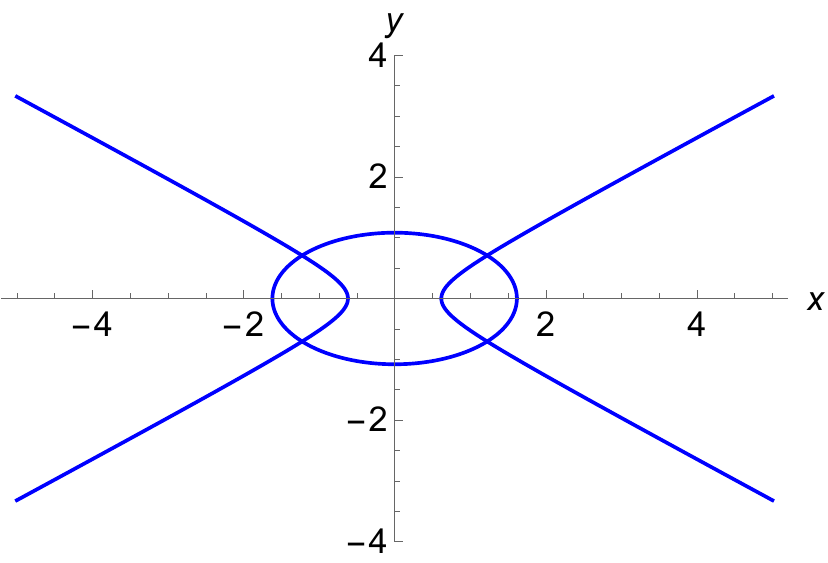}}
	\caption{Trajectories of the Hamiltonian $H_{V}$ (\ref{HV}) obtained from the algebraic equation (\ref{Tr5}). The case (a) correspond to the values $H_{V}=4,\,L=2,\,A=2$, whereas in (b) $H_{V}=4,\,L=2,\,A=0$. }
	\label{H5F}
\end{figure*}

\subsection{A non-separable Hamiltonian}

In \cite{Post_2011}, it was introduced a 2D superintegrable Hamiltonian system that does not allow additive separation of variables in the Hamilton-Jacobi equation in any system of coordinates. The corresponding Hamiltonian is of the form:

\begin{equation}
  H_{VI}\ = \ \frac12(p_x^2+p_y^2)\ + \ \frac{\alpha \,y}{x^{\frac23}}  \ .
\end{equation}

\noindent It admits two algebraically independent integrals, a cubic integral in the momentum variables

\begin{equation}
    L \ = \ 3\,p_x^2\,p_y \ + \ 2\,p_y^3 \ + \ 9\,\alpha\, x^{\frac13}\, p_x \ + \ \frac{6\,\alpha\, y}{x^{\frac23}} \,p_y \ ,
\end{equation}

\noindent and the quartic integral

\begin{equation}
    A \ = \ p_x^4 \ + \ \frac{4\,\alpha\, y}{x^{\frac23}}\,p_x^2 \ - \ 12\,x^{\frac13}\,\alpha \,p_x\,p_y\ - \ \frac{2\,\alpha^2\,(9x^2-2y^2)} {x^{\frac43}} \ .
\end{equation}

As indicated in \cite{Post_2011}, we have that $B \,=\,\{L,A\} \ = \ 108 $, hence $k=0$ in (\ref{eqbonatsos}). This implies that the polynomial algebra generated by the four elements $(H=H_{VI},L,A,B)$ is trivially a linear one.

\subsection{A Hamiltonian trigonometric on the momentum coordinates}

To illustrate the general applicability of the present construction, let us consider the following Hamiltonian function

\begin{equation}
H_{VII}\ = \ \cos(y\,p_{y}^2).
\end{equation}
The two functions
\begin{equation}
L \   = \ p_{x}\quad\text{and}\quad A\ = \ (\sin x)\left(p_{x}^{2}\cos^{3}(y\,p_{y}^2)\right)\ ,
\end{equation}
are constants of motion such that
\begin{equation}
\lbrace L,\lbrace L,A\rbrace\rbrace\ = \ -A\,;
\end{equation}
Indeed, 
\begin{equation}
\lbrace L,A\rbrace\ = \ -\cos x\left( p_{x}^{2}\cos^{3}(y\,p_{y}^2) \right) \ ,
\end{equation}
and
\begin{equation}
\lbrace L,\lbrace L,A\rbrace\rbrace=-\sin x\left(p_{x}^{2}\cos^{3}(y\,p_{y}^2)\right)\ = \ -A\ .
\end{equation}
On the other hand
\begin{equation}
B^{2}+A^{2}\ = \ p_{x}^{4}\cos^{6}(yp_{y}^{2})\ = \ L^{4}\,H^{6}\quad , \qquad H=H_{VII} \ ,
\end{equation}
as well as
\begin{equation}
\lbrace A,B\rbrace \ = \ -2\,L^{3}\,H^{6}\ .
\end{equation}
Thus, even in the case when the superintegrable Hamiltonian is not of the standard form (i.e., a quadratic function in the momentum variables) it is possible to construct the associated polynomial algebra of integrals.

\subsection{The harmonic oscillator in a curved space with constant curvature}

The harmonic oscillator in a space with constant curvature has been studied in the classical case in \cite{higgs1979dynamical} and in the quantum case in \cite{bonatsos1994deformed}. The Hamiltonian function is 
\begin{equation}
H_{VIII} \ = \ \frac{1}{2}(\pi_{x}^{2} + \pi_{y}^{2} + \lambda L^{2}) + \frac{\omega^{2}}{2}(x^{2} + y^{2})\ ,
\end{equation}
where $L$ is the angular momentum
\begin{equation}
L \  = \  x\,p_{y} - y\,p_{x}\ ,
\end{equation}
and
\begin{equation}
\begin{split}
\pi_{x} &= p_{x} + \frac{\lambda}{2}(x(xp_{x} + yp_{y}) + (xp_{x} + yp_{y})x)\ , \\
\pi_{y} &= p_{y} + \frac{\lambda}{2}(y(xp_{x} + yp_{y}) + (xp_{x} + yp_{y})y)\ . 
\end{split}
\end{equation}
We have that $L$ and the so called Fradkin-like tensor
\begin{equation}\label{eq4}
A \ = \ S_{xx}-S_{yy}=(\pi_{x}^{2} + \omega^{2}x^{2}) - (\pi_{y}^{2} + \omega^{2}y^{2}) \ ,
\end{equation}
are constants of motion.

\begin{equation}
\begin{split}
\lbrace L, A \rbrace &= 4 (p_{x}^{2}\lambda x(1 + \lambda x^{2})y + p_{x}p_{y}(1 + 2\lambda^{2} x^{2}y^{2} + \lambda(x^{2} + y^{2})) + \\
&\quad + xy(\omega^{2} + p_{y}^{2} \lambda(1 + \lambda y^{2}))), \\
 &= B
\end{split}
\end{equation}
and
\begin{equation}
\begin{split}
\lbrace L, B \rbrace &= -4\bigg(\omega^{2}(x^{2} - y^{2}) + 2p_{x}p_{y}\lambda^{2} xy(x^{2} - y^{2}) +\bigg. \\
&\quad + \bigg. p_{x}^{2}(1 + 2\lambda x^{2} + \lambda^{2}(x^{4} - x^{2}y^{2}))\bigg. - \\ &\quad - \bigg. p_{y}^{2}(1 + 2\lambda y^{2} + \lambda^{2}(-x^{2}y^{2} + y^{4}))\bigg)\\
&= -4A;
\end{split}
\end{equation}
so $k=2$ and
\begin{equation}
G(H,L)=B^{2} + 4A^{2}=4\lambda^{2}L^{4} - 16\lambda HL^{2} - 4\,\omega^{2}L^{2} + 4\,H^{2},
\end{equation}
therefore
\begin{equation}
\begin{split}
\lbrace A,B \rbrace &= - \frac{1}{2}\frac{\partial G}{\partial L}, \\
&= -8\,\lambda^{2}\,L^{3} + 16\,\lambda\, H\,L + 4\,\omega^{2}L\ ,
\end{split}
\end{equation}
here $H=H_{VIII}$. At $\lambda=0$ we recover (\ref{ABI}), as it should be.

\vspace{0.2cm}

\noindent Table~\ref{tab:algebra_types} presents a summary of the polynomial algebra types associated with the integrals of motion for the two-dimensional superintegrable systems studied.

\begin{table}[h!]
\centering
\renewcommand{\arraystretch}{1.2}
\small
\begin{tabular}{|l|c|c|p{7.8cm}|}
\hline
\textbf{System} & \textbf{\(H\)} & \textbf{Type} & \textbf{Remarks} \\
\hline
Isotropic Harmonic Oscillator & \( H_I \) & Linear & Standard example; Poisson algebra closes linearly as \( \{A, B\} \propto L \); \( G = (2\,m\,H)^{2}\, - \, \omega^{2}\,m^{2}\,(2\,L)^{2}  \). \\
\hline
Kepler Problem & \( H_{II} \) & Quadratic & Hidden symmetry via Laplace–Runge–Lenz vector; \( \{A, B\} \propto LH \); \( G = \alpha^2 +2\,H\,L^2 \). \\
\hline
Fokas–Lagerstrom System & \( H_{III} \) & Cubic & Nontrivial cubic closure; multibranch trajectories; \( G = -4L^4 + 32H^3L - 48H^2L^2 + 24HL^3 \). \\
\hline
Holt System & \( H_{IV} \) & Quadratic & Inverse-square singularity; \( \{A, B\} \) quadratic in \( H, L \); \( G = 32\,L\,(2H -L)^{2} \ - \ 512\,\delta\, L \). \\
\hline
Smorodinsky–Winternitz & \( H_V \) & Cubic & Cubic algebra reducible to quadratic form;  \( G = 256\left(\left(H - \frac{L}{2}\right)^{2} - 4bc\right)L^{2}\). \\
\hline
Post–Winternitz (non-sep.) & \( H_{VI} \) & Linear & Nonseparable; algebra closes with constant bracket \( \{L, A\} \); \( G = \text{const.\,} \). \\
\hline
Trigonometric Momentum & \( H_{VII} \) & Quadratic & Deformed oscillator-type algebra; \( \{L, \{L, A\}\} = -A \); \( G =  L^{4}\,H^{6} \). \\
\hline
Oscillator on Curved Space & \( H_{VIII} \) & Quadratic & Curvature-dependent algebra; structure function involves \( \lambda \); \( G =4\lambda^{2}L^{4} - 16\lambda HL^{2} - 4\omega^{2}L^{2} + 4H^{2} \). \\
\hline
\end{tabular}
\caption{Type of the polynomial algebra of integrals in various 2D superintegrable systems, including the associated structure function \( G(H, L) \).}
\label{tab:algebra_types}
\end{table}

\section{Conclusions}

In this work, a rigorous and constructive proof has been presented establishing the existence of a four-dimensional polynomial Poisson algebra for classical two-dimensional superintegrable Hamiltonian systems. This construction is fully general: it does not assume that the integrals of motion are polynomial in the canonical momenta, nor does it rely on separability of the Hamilton–Jacobi equation. The main result shows that, for any classical superintegrable system in two dimensions, it is always possible to construct two functionally independent integrals of motion whose Poisson brackets close into a structure analogous to a deformed oscillator algebra.\\

Several physically and mathematically relevant examples were presented to illustrate the generality of the result, including well-known systems such as the harmonic oscillator, the Kepler problem, the Holt, Smorodinsky–Winternitz, and Fokas–Lagerstrom potentials, as well as non-separable models. In each case, the explicit polynomial algebra among the integrals of motion was derived, and in many instances, the trajectories in configuration space were obtained through purely algebraic manipulations, without solving differential equations.\\

A particularly noteworthy feature of the method is its applicability to systems with integrals that are not polynomial in the momenta—an aspect largely unexplored in existing literature. The framework also accommodates systems defined on curved spaces, further demonstrating the versatility of the approach.\\

An important structural feature emerging from our analysis is the local \emph{abelianization} of the integrals of motion on specific invariant submanifolds, where originally noncommuting conserved quantities become involutive. This phenomenon highlights distinguished dynamical sectors in which the system exhibits a local Liouville-integrable behavior, despite its global superintegrable character. Furthermore, the study of the algebraic equations obtained from the integrals of motion, used to derive trajectories in configuration space without solving the equations of motion, reveals that the \emph{discriminant} of these equations carries meaningful physical and geometric information. In particular, it would identify critical values of the integrals that correspond to bifurcations in trajectory topology or the onset of degeneracies. It will be done elsewhere. \\

These findings provide a deeper understanding of the phase space structure of superintegrable systems and suggest new avenues for algebraic orbit classification, symbolic integration, and applications to the spectral theory of quantum analogues.

\section*{Acknowledgments}

The author R. Azuaje thanks Costas Daskaloyannis for some helpful comments during the initial stages of this study. A. M. Escobar Ruiz gratefully acknowledges A. Turbiner for valuable discussions on polynomial algebras in both quantum and classical settings. R. Azuaje thanks the financial support provided by the Secretaría de Ciencia, Humanidades, Tecnología e Innovación (SECIHTI) of Mexico through a postdoctoral fellowship under the Estancias Posdoctorales por México 2022 program. A.M. Escobar Ruiz acknowledges support from the Consejo Nacional de Humanidades, Ciencias y Tecnologías (CONAHCyT) of Mexico under Grant No. CF-2023-I-1496, as well as institutional funding from Universidad Autónoma Metropolitana (UAM) through research grant 2024-CPIR-0.

\bibliography{refs} 
\bibliographystyle{unsrt} 

\end{document}